\documentclass[11pt]{article}

\usepackage[utf8]{inputenc}
\usepackage[T1]{fontenc}
\usepackage{graphicx}
\usepackage{amssymb}
\usepackage{amsmath}
\usepackage{xcolor}
\usepackage{tcolorbox}
\usepackage{soul}
\usepackage{amsthm}
\usepackage{hyperref}
\usepackage{booktabs}
\usepackage{makecell}
\usepackage{authblk} 
\usepackage[authoryear]{natbib}
\theoremstyle{plain}
\newtheorem{lemma}{Lemma}
\newtheorem{proposition}{Proposition}

\newtheorem{theorem}{Theorem}
\newtheorem{corollary}{Corollary}

\theoremstyle{definition}
\newtheorem{example}{Example}
\newtheorem{definition}{Definition}

\theoremstyle{remark}

\DeclareMathOperator*{\argmin}{arg\,min}

\newcommand{\mset}[1]{\{\!\!\{ #1 \}\!\!\}}

\def\Sigma{{m}}

\title{Consistency, unanimity, and the Borda rule in social ranking}

\author[1]{Takahiro Suzuki\thanks{Corresponding author: \texttt{suzuki-tkenmgt@g.ecc.u-tokyo.ac.jp}}}
\author[2]{Stefano Moretti}
\author[2,3]{Rachel Ruell\'{e}}

\affil[1]{Department of Civil Engineering, Graduate School of Engineering, The University of Tokyo, Tokyo, Japan}
\affil[2]{LAMSADE, CNRS, Université Paris-Dauphine, Université PSL, Paris, France}
\affil[3]{ENS Rennes, Bruz, France}

\date{}

\begin{document}

\maketitle

\begin{abstract}
The social ranking is a recently proposed framework for evaluating the power of individuals according to the performance ranking of their coalitions. Although its origin can be traced to the classical power indices in simple games, social ranking approaches carry out this evaluation within the ordinal framework of social choice theory. This article introduces the Borda rule into social ranking. Specifically, we focus on two essential properties of the Borda rule---consistency and closeness to unanimity---and investigate the social ranking solutions (SRSs) satisfying these properties. Among several possible definitions of the Borda rule as an SRS, we characterize one of such solutions by (a weak version of) consistency, closeness to unanimity (under the linear and symmetric domain), neutrality (i.e., names of the individuals do not matter), and independence of perfunctory participation (i.e., adding a perfunctory coalition into the worst class of the coalitional ranking does not affect the social ranking). We therefore propose a new Borda-type SRS for evaluating the competence of individuals in coalitional contexts.
\end{abstract}

\section{Introduction}
There are many situations in which the importance or competence of individuals, such as athletes in team sports \citep{Algaba2021} or features in explainable AI \citep{Gourves2025}, is judged according to the performance ranking of their coalitions. The origins of this problem can be traced to the classical theory on power indices, such as the Shapley--Shubik power index \citep{shapley1954method} in simple games, but \citet{Moretti2015} and \citet{Moretti2017} have recently established the ordinal version of this problem, called the \textit{social ranking problem}. Theoretically, this is considered an ordinal counterpart to the measurement of power indices using the ordinal approach of social choice theory. 

The theory of the social ranking problem has developed rapidly over the past decade. Several rational social ranking solutions (SRSs) have been proposed axiomatically: the \textit{ceteris paribus} (CP) majority \citep{Haret2019a,Allouche2021,Suzuki2026a}, the ordinal Banzhaf index \citep{Khani2019}, the lexicographic excellence solution and its variants \citep{Bernardi2019,Algaba2021,Beal2022a,Suzuki2024c,Suzuki2024,Aleandri2024a,Suzuki}, plurality \citep{Suzuki2025_Mill_SRS}, the intersection initial segment \citep{Suzuki2026b}, and CP-Borda \citep{Ruelle2026}. The purpose of this article is to introduce a new and arguably important SRS: the Borda rule. 

The Borda rule is one of the central objects in social choice theory. Since \citet{Young1974a} characterized it by neutrality, consistency, faithfulness, and cancellation, many subsequent studies have axiomatically characterized the Borda rule \citep{Hansson1976,Nitzan1981a,Marchant2000,Heckelman2020} and investigated its properties and its counterparts in other models \citep{Black1976,Tabarrok1999,Ohseto2007,Suzuki2017b,Garcia-Lapresta2025}. 

The central difficulty in incorporating the Borda rule into the theory of SRSs is the ambiguity of its definitions. Usually, an SRS is defined as a function that maps a weak order of coalitions to a weak order of individuals. Although the ambiguity of the Borda rule in the presence of indifferences is well-known in the standard voting contexts \citep[see, e.g.,][]{Terzopoulou2021}, we also face the problem that each individual can appear several times in a single ranking (because they can belong to multiple coalitions). In fact, \citet{Suzuki2024} give two possible definitions of the Borda rule as an SRS, and \citet{Ruelle2026}  define what they call the CP-Borda rule on the basis of CP comparisons of individuals. Our strategy is to focus on some essential properties of the Borda rule and establish a corresponding SRS based on these properties. 

In particular, we focus on \textit{consistency} and \textit{closeness to unanimity}. Consistency demands that the social judgments from two disjoint sub-societies are ``consistent'' with that of the unified society; this is a familiar axiom in the characterization of the Borda rule \citep{Young1974a,Nitzan1981a} and scoring rules \citep{H.P.Young1975}. As for closeness to unanimity, \citet{Farkas1979} demonstrate that the winner(s) selected by the Borda rule are closest to unanimity in the sense that they are the set of alternatives that can reach unanimity (the top alternative of each voter) with the minimum number of swaps between neighboring alternatives. 

In this paper, the Borda SRS is defined such that (i)~the score of each coalition $S$ is the number of equivalence classes ranked lower than $S$,\footnote{Notice that under the standard Borda rule, an alternative's score is given by the number of alternatives ranked lower than it. The statement~(i) is obtained simply by replacing the lower alternatives with the lower equivalence classes. A discerning reader might wonder why we count lower equivalence classes rather than lower coalitions. Both of these will be formalized in Definition~\ref{definition:Borda_rules}, but it will be demonstrated that the former is better suited to the two fundamental properties.} (ii)~the score of each individual $i$ is the sum of the scores assigned to each of their coalitions, and (iii)~individuals with higher scores are ranked higher. Our main result (Theorem~\ref{theorem: characterization of B1}) characterizes this Borda SRS, denoted as $R^{B1}$, in terms of two fundamental axioms, (a weak version of) consistency and closeness to the unanimity, and two complementary axioms, neutrality (i.e., the names of the alternatives do not matter) and independence of perfunctory participation (i.e., adding a ``perfunctory'' coalition into the worst class does not affect the social ranking). A close relationship with the Banzhaf value in coalitional games will also be revealed (Proposition~\ref{proposition:banzhaf}).

The remainder of this paper is structured as follows. Section~\ref{section:model} presents the basic model, followed by possible definitions of the Borda rule as SRSs. In Section~\ref{section:preliminary result}, we define four possible definitions of Borda-like SRSs, each of which is considered ``closest to unanimity'' in some linear and symmetric domain. Our main result in Section~\ref{section:main result} characterizes one of the possible definitions using (a certain type of) consistency, the property of being closest to unanimity, and supplementary axioms. Concluding remarks are given in Section~\ref{section:conclusion}.

\section{Model}
\label{section:model}
\subsection{Basic model}
Let $X$ be the set of individuals, and let $\mathcal{X}=2^{X}$ be the set of all coalitions. For a set $A$, let $\mathcal{R}(A)$ be the set of all weak orders (i.e., reflexive, complete, and transitive binary relations) on $A$. We denote by $\mathbb{N}=\{0,1,2,\dots\}$ the set of all nonnegative integers and by $\mathbb{N}_{+}=\{1,2,\dots\}$ the set of all positive integers. For $n\in\mathbb{N}$, let $[n]:=\{k\in\mathbb{N}\mid 1\leq k\leq n\}$.  

A \textit{multiset} on a set $A$ is a function $m:A\rightarrow\mathbb{N}$. The set of all multisets on $A$ is denoted as $\mathcal{M}(A)$. For $a\in A$, $m(a)$ (the \textit{multiplicity} of $a$) indicates how many times $a$ appears in the multiset. We often denote a multiset by listing each element $a\in A$ a total of $m(a)$ times within $\mset{}$. For instance, if $A=\{a,b,c\}$, $m(a)=m(b)=2$, and $m(c)=1$, then we write $m=\mset{a,a,b,b,c}$. Two special multisets, $\emptyset$ and $1_{a}$ ($a\in A$), are defined as follows: $\emptyset (a)=0$ for all $a\in A$, and 
\begin{align*}
1_{a}(x)=\begin{cases}
1\quad&\text{if }x=a,\\
0\quad&\text{otherwise.}
\end{cases}
\end{align*}
For $m\in \mathcal{M}(A)$ and $a\in A$, with a slight abuse of notation, we write $a\in {m}$ if $m(a)>0$. For $m,m'\in M(A)$, $m+m'$ and $m\setminus m'$ are defined as follows: 
\begin{align*}
&(m+m')(a):=m(a)+m'(a)\quad \forall a\in A,\\
&(m\setminus m')(a):= \max\{0,m(a)-m'(a)\}.
\end{align*}

For $l\in\mathbb{N}_{+}$, let $\mathcal{R}_{l}:=\{(\Sigma_{l},\Sigma_{l-1},\dots,\Sigma_{1})\mid \Sigma_{k}\in\mathcal{M}(\mathcal{X})\setminus\{\emptyset\}\quad\text{for all }k\in[l]\}$ and $\mathcal{R}:=\bigcup_{l\in\mathbb{N}}\mathcal{R}_{l}$. An element $\mathord{\succsim}=(\Sigma_{l},\Sigma_{l-1},\dots,\Sigma_{1}) \in \mathcal{R}_{l}$ is called a \textit{coalitional ranking}, where the multiset $\Sigma_{k}$ is interpreted as the $k$-th-worst ``equivalence class'' in which each coalition $S\in \Sigma_k$ is represented with multiplicity $m_k(S)>0$. For $k>k'$, each element of $\Sigma_{k}$ is considered to be more competitive than that of $\Sigma_{k'}$. For $l\in \mathbb{N}_{+}$, let $\mathcal{L}_{l}:=\{(\Sigma_{l},\dots,\Sigma_{1})\in \mathcal{R}_{l}\mid \forall k\in[l],\;\lvert \Sigma_{k}\rvert=1\}$. Furthermore, for $x,y\in X$, we write $\mathcal{R}^{x,y}_{l}:=\{(\Sigma_{l},\dots,\Sigma_{1})\in\mathcal{R}_{l}\mid \sum_{k\in[l]}\sum_{S\in \Sigma_k: x \in S}\Sigma_{k}(S)=\sum_{k\in[l]}\sum_{T\in \Sigma_k: y \in T}\Sigma_{k}(T)\}$. In words, $\mathcal{R}^{x,y}_{l}$ is the set of all coalitional rankings in $\mathcal{R}_{l}$ in which the number of coalitions including $x$ and the number including $y$ are exactly the same. Similarly, we write $\mathcal{R}^{x,y}:=\bigcup_{l\in\mathbb{N}_{+}}\mathcal{R}^{x,y}_{l}$, $\mathcal{L}^{x,y}_{l}:=\mathcal{R}^{x,y}\cap\mathcal{L}_{l}$, and $\mathcal{L}^{x,y}:=\bigcup_{l\in \mathbb{N}_{+}}\mathcal{L}^{x,y}_{l}$. 

\begin{example}
\label{example:journals}
Suppose that there are three academic journals J1, J2, and J3 in some field. In general, J1 is ranked highest, followed by J2, and then J3. Suppose that Alice and Bob published their joint work twice in J1 and once in J2, Bob and Charlie published their joint work once in J1 and once in J3, Alice and Charlie each published one single-authored work in J3, Bob published two single-authored papers in J3, and Charlie published one single-authored paper in J2. This situation is expressed as follows: 
$\mathord{\succsim}=(\mset{\{\text{A},\text{B}\},\{\text{A},\text{B}\},\{\text{B},\text{C}\}},\mset{\{\text{A},\text{B}\},\{\text{C}\}},\mset{\{\text{A}\},\{\text{C}\},\{\text{B},\text{C}\},\{\text{B}\},\{\text{B}\}})$, where A is Alice, B is Bob, and C is Charlie. By definition, we have $\mathord{\succsim}\notin\mathcal{R}^{\text{A},\text{B}}_{3}$ (because Alice published fewer papers overall than Bob), $\mathord{\succsim}\notin \mathcal{L}_{3}$ (because the equivalence classes are not singletons), and $\mathord{\succsim}\in\mathcal{R}^{\text{A},\text{C}}_{3}$. 
\end{example}

For a set $A$, let $\mathcal{B}(A)$ denote the set of all reflexive and complete binary relations on $A$. For a binary relation $R$ on a set $A$, the restriction of $R$ into $B\subseteq A$ is denoted as $R\mid_{B}$. 

\begin{definition}
A \textit{social ranking solution} (SRS) is a function $R:\mathcal{R}\rightarrow\mathcal{B}(X)$. 
\end{definition}

We denote $R(\succsim)$ as $R_{\succsim}$. As usual, the asymmetric and symmetric parts of $R_{\succsim}$ are denoted as $P_{\succsim}$ and $I_{\succsim}$, respectively.

\subsection{Axioms}
For a permutation $\sigma:X\rightarrow X$ and $\Sigma\in \mathcal{M}(\mathcal{X})$, let $\Sigma_{\sigma}\in\mathcal{M}(\mathcal{\mathcal{X}})$ be such that $\Sigma^{\sigma}(S):=\Sigma(\sigma(S))$, where $\sigma(S)=\{\sigma(x)\mid x\in S\}$. Intuitively, $\Sigma^{\sigma}$ is a multiset obtained from $\Sigma$ by replacing each $x\in X$ with $\sigma^{-1}(x)$. For $\mathord{\succsim}=(\Sigma_{l},\dots,\Sigma_{1})\in\mathcal{R}_{l}$, let $\mathord{\succsim}^{\sigma}:=(\Sigma_{l}^{\sigma},\dots,\Sigma_{1}^{\sigma})$. 

\begin{definition}
An SRS $R$ satisfies \textit{neutrality} (NT) if for any $\mathord{\succsim}\in\mathcal{R}$, $x,y\in X$, and a permutation $\sigma$ on $X$, we have $xR_{\succsim^{\sigma}}y\Rightarrow \sigma(x)R_{\succsim}\sigma(y)$.
\end{definition}

NT is a standard axiom demanding that the names of the individuals do not matter. 
For $l\in\mathbb{N}_{+}$, $\mathord{\succsim}=(\Sigma_{l},\dots,\Sigma_{1})\in\mathcal{R}_{l}$, and $\mathord{\succsim}'=(\Sigma'_{l},\dots,\Sigma'_{1})\in\mathcal{R}_{l}$, let $\mathord{\succsim}+\mathord{\succsim}'\in\mathcal{R}_{l}$ be such that 
$\mathord{\succsim}+\mathord{\succsim}':=(\Sigma_{l}+\Sigma'_{l},\dots,\Sigma'_{1}+\Sigma_{1})$.

\begin{definition}
An SRS $R$ satisfies \textit{equi-size consistency} (ECON) if for any $l\in\mathbb{N}_{+}$, $\mathord{\succsim},\mathord{\succsim}'\in\mathcal{R}_{l}$, and $x,y\in X$, we have that 
\begin{enumerate}
\item if $xR_{\succsim}y$ and $xR_{\succsim'}y$, then $xR_{\succsim+\succsim'}y$, and 
\item if $xP_{\succsim}y$ and $xR_{\succsim'}y$, then $xP_{\succsim+\succsim'}y$.
 \end{enumerate}
\end{definition}

In a study on social ranking, \citet{Suzuki2024} define consistency as the condition that the above conditions 1 and~2 hold for \textit{any} combination of two disjoint coalitional orders $\succsim$ and $\succsim'$ (the scope of their definition includes cases in which $\succsim$ and $\succsim'$ have different numbers of classes, as well as those in which $\succsim$ and $\succsim'$ are concatenated somehow).  As a result, they show that some of the possible definitions of the Borda rule do not satisfy consistency. Our ECON is a weaker form of consistency: it requires only that the judgment for a merged coalitional order be consistent with those for two underlying coalitional orders, provided that they have the same number of classes and that each pair of corresponding classes is merged. 

For $\mathord{\succsim}=(\Sigma_{l},\dots,\Sigma_{1})\in\mathcal{R}_{l}$, $k\in[l]$, and $m\in\mathcal{M}(\mathcal{X})$, let $\mathord{\succsim}+(m)_{k}$ be a coalitional ranking that is obtained from $\succsim$ by adding $m$ to the $k$-th-worst equivalence class, that is, 
$\mathord{\succsim}+(m)_{k}:=(\Sigma_{l}',\dots,\Sigma'_{1})\in\mathcal{R}_{l}$, where 
\[\Sigma'_{j}=\begin{cases}
\Sigma_{j}+m\quad&\text{if }j=k,\\
\Sigma_{j}\quad&\text{otherwise.}
\end{cases}
\]

\begin{definition}
An SRS $R$ satisfies \textit{independence of perfunctory participation} (IPP) if for any $\mathord{\succsim}\in\mathcal{R}_{l}$, $x,y\in X$, and $S\in \mathcal{X}$ such that $x\in S\not\ni y$, we have that $R_{\succsim}\mid_{\{x,y\}}=R_{\succsim+(1_{S})_{1}}\mid_{\{x,y\}}$. 
\end{definition}

IPP demands that adding $S$, with $x\in S\not\ni y$, to the worst class of the coalitional ranking does not affect the social ranking between $x$ and $y$. One straightforward interpretation is that the publication of the lowest-quality work (such as perfunctory papers published in a predatory journal) does not alter one's evaluation. 

\begin{definition}
An SRS $R$ satisfies \textit{discrete continuity} (DCONT) if for any $l\in \mathbb{N}$ with $l\geq 2$, $k_{1}\in[l-1]$, $\mathord{\succsim}=(\Sigma_{l},\dots,\Sigma_{1})\in\mathcal{R}_{l}$, $x,y\in X$, and $S\in\mathcal{X}$, if $xP_{\succsim+(\mset{S})_{k_{1}}}y$, then $\neg (yP_{\succsim+(\mset{S})_{k_{1}+1}}x)$. 
\end{definition}
DCONT demands that the social ranking can change only ``continuously'': a shift of $S$ to a neighboring class does not completely reverse the social ranking between two individuals. It is interesting to consider a logically equivalent definition of DCONT. Let us say that $R$ satisfies DCONT* if for any $l\in \mathbb{N}$ with $l\geq 2$, $k_{1}\in[l-1]$, $\mathord{\succsim}=(\Sigma_{l},\dots,\Sigma_{1})\in\mathcal{R}_{l}$, $x,y\in X$, and $\Gamma\in \mathcal{M}(\mathcal{X})$, if $xP_{\succsim+(\Gamma )_{k_{1}}}y$ and $yP_{\succsim+(\Gamma)_{k_1+1}}x$, then there exist $\Gamma_{1},\Gamma_{2}\in\mathcal{M}(\mathcal{X})$ with $\Gamma=\Gamma _{1}+\Gamma _{2}$ such that $xI_{\succsim+(\Gamma_{1})_{k_{1}}+(\Gamma_{2})_{k_{1}+1}}y$. In other words, if moving $\Gamma$ to the neighboring class (the $k$-th class to the $(k+1)$-th class) reverses the social ranking between $x$ and $y$, then there must exist $\Gamma_{1}$ (which is a subset of $\Gamma$) such that the shift of $\Gamma$ makes $x$ and $y$ indifferent. We can verify that DCONT and DCONT* are logically equivalent as follows. 

To show that DCONT$\Rightarrow$DCONT*, let $\Gamma=\mset{S_{1},\dots,S_{m}}$ and suppose that $xP_{\succsim+(\Gamma)_{k_{1}}}y$ and $yP_{\succsim+(\Gamma)_{k_1+1}}x$. We will prove the existence of $\Gamma_{1}$ and $\Gamma_{2}$  with $\Gamma=\Gamma _{1}+\Gamma _{2}$ such that $xI_{\succsim+(\Gamma_{1})_{k_{1}}+(\Gamma_{2})_{k_{1}+1}}y$.
For each $\mu=1,2,\dots,m$, let $\mathord{\succsim}^{(0)}:=\mathord{\succsim}+(\Gamma)_{k_{1}}$, and let $\mathord{\succsim}^{(\mu)}$ be the coalitional ranking obtained from $\mathord{\succsim}^{(\mu-1)}$ by shifting $S_{\mu}$ from the $k_{1}$-th class to the $(k_{1}+1)$-th class, that is, $\mathord{\succsim}^{(\mu)}:=\mathord{\succsim}+(\mset{S_{1},\dots,S_{\mu}})_{k_{1}+1}+(\mset{S_{\mu+1},\dots,S_{m}})_{k_{1}}$. By assumption, we have $xP_{\succsim^{(0)}}y$. Thus, DCONT (under the assumption that $R_{\mathord{\succsim^{(1)}}}$ is complete) implies that either $xP_{\succsim^{(1)}}y$ or $xI_{\succsim^{(1)}}y$. If the latter holds, the proof can be completed by letting $\Gamma_{1}=\mset{S_{1}}$. Assume the former. Then, we can apply DCONT again. It follows that either $xP_{\succsim^{(2)}}y$ or $xI_{\succsim^{(2)}}y$. If the latter holds, the proof is completed (by letting $\Gamma_{1}=\mset{S_{1},S_{2}}$). Assume the former. We can repeat this process to find a desired $\Gamma_{1}$. Notice that this process cannot continue until $\mu=m-1$, because in such a case, DCONT implies either $xP_{\succsim^{(m)}}y$ or $xI_{\succsim^{(m)}}y$. However, $\mathord{\succsim}^{(m)}=\mathord{\succsim}+(\Gamma)_{k_{1}+1}$ by the definition. So, we have $yP_{\succsim^{(m)}}x$ by assumption. This means that we must encounter a desired $\Gamma_{1}$ within the above process. 

For DCONT*$\Rightarrow$DCONT, assume that $xP_{\succsim+(\mset{S})_{k_{1}}}y$. Let $\Gamma:=\mset{S}$ in the definition of DCONT*. Then, it follows that there must exist $\Gamma_{1},\Gamma_{2}\in\mathcal{M}(\mathcal{X})$ with $\Gamma=\Gamma_{1}+\Gamma_{2}$ such that $xI_{\succsim+(\Gamma_{1})_{k_{1}+1}+(\Gamma_{2})_{k_{1}}}y$. Since $\Gamma$ is a singleton, we have either $(\Gamma_{1},\Gamma_{2})=(\Gamma,\emptyset)$ or $(\Gamma_{1},\Gamma_{2})=(\emptyset,\Gamma)$. However, the latter is inappropriate because it means that no shift is made, so it cannot hold that $xP_{\succsim+(\Gamma)_{k_{1}}}y$ and $xI_{\succsim+(\Gamma_{1})_{k_{1}+1}+(\Gamma_{2})_{k_{1}}}y$. Hence, $(\Gamma_{1},\Gamma_{2})=(\Gamma,\emptyset)$ and $\neg (yP_{\succsim+(\mset{S})_{k_{1}+1}}x)$. This proves DCONT.

\section{Preliminary results on alternative definitions of the Borda rule}
\label{section:preliminary result}

\subsection{The Borda rule and the inversion number}
\label{subsection:Borda rule and the inversion number}
One of the noteworthy properties of the Borda rule is that its winners are the closest to unanimous agreement. \citet{Farkas1979} show that the alternatives selected by the Borda rule can become the most preferred alternatives by everyone with the minimum number of swaps. For instance, consider the following preference profile: 
\begin{align*}
&\mathord{\succsim}_{1}:b,a,c;\\
&\mathord{\succsim}_{2}:b,c,a;\\
&\mathord{\succsim}_{3}:c,b,a;\\
&\mathord{\succsim}_{4}:a,c,b.
\end{align*}
A \textit{swap} is an interchange in the ranking of two neighboring alternatives (e.g., $a$ and $c$ in $\succsim_{1}$). In the above case, $\succsim_{1}$ can be transformed to $\succsim_{2}$ with a single swap (i.e., interchanging $a$ and $c$), $\succsim_{3}$ can be transformed to $\succsim_{4}$ with two swaps (interchanging the positions of $a$ and $b$ and then those of $a$ and $c$). Farkas and Nitzan show that the winners of the Borda rule are equivalent to $\argmin_{x} s(x)$, where $s(x)$ represents the number of swaps required to move $x$ to the top of everyone's preferences.

The purpose of this subsection is to establish a parallel argument in our social ranking framework. For $\mathord{\succsim}=(\Sigma_{l},\dots,\Sigma_{1}) \in\mathcal{L}_{l}$ and $x\in X$, we say that $x$ is \textit{unanimously supported} if  there exists $k\in[l]$ such that each $\Sigma_{l},\dots,\Sigma_{k}$ is made up of only coalitions including $x$, and each $\Sigma_{k-1},\dots,\Sigma_{1}$ is made up of only coalitions not including $x$. We denote by $\hat{i}_{\succsim}(x)$ the minimum number of swaps that need to be performed on $\succsim$ for $x$ to become unanimously supported.

\begin{proposition}
\label{proposition:inversion_formula}
Let $l\in\mathbb{N}_{+}$. For each $\mathord{\succsim}\in\mathcal{L}_{l}$ and $x\in X$, $\hat{i}_{\succsim}(x)$ is given by:
$\hat{i}_{\succsim}(x)=\sum_{k\in[l]}\sum_{S\in \Sigma_k: x \in S}\Sigma_{k}(S)\cdot \sum_{k'>k}\sum_{T\in \Sigma_{k'}: x \notin T}\Sigma_{k'}(T)$.
\end{proposition}

\begin{proof}[Sketch of proof]
Let $l\in\mathbb{N}_{+}$. Notice that any element $\mathord{\succsim}\in \mathcal{L}_{l}$ can be written as $\mathord{\succsim}=(\mset{S_{l}},\dots,\mset{S_{1}})$ for some $S_{l},\dots,S_{1}\in\mathcal{X}$. For $x\in X$ and $\mathord{\succsim}=(\mset{S_{l}},\dots,\mset{S_{1}})$, let $x^{\succsim}:=(x^{\succsim}_{l},\dots,x^{\succsim}_{1})$, where 
\[x^{\succsim}_{k}:=\begin{cases}
1\quad\text{if $x\in S_{k}$},\\
0\quad\text{if $x\notin S_{k}$}.
\end{cases}
\]
Then, swapping the neighboring coalitions at $\succsim$ to make $x$ unanimously supported is equivalent to sorting the vector $x^{\succsim}$ in descending order by swapping the neighboring elements of the vectors. The total number of swaps required for the sorting problem is known as the \textit{inversion number} in algorithmic theory \citep[see, e.g.,][]{Dinneen2016}. In this case, it is expressed by the equation in the proposition. 
\end{proof}

Proposition~\ref{proposition:inversion_formula} shows that the closeness to unanimity for a linear order can be measured by the inversion number in algorithmic theory. Based on this observation, we propose several definitions of Borda rule as an SRS in section~\ref{subsection:four definitions of borda rule as an SRS}, and then proceed to an axiomatic analysis in Section~\ref{section:main result}.

\subsection{Four definitions of the Borda rule as an SRS}
\label{subsection:four definitions of borda rule as an SRS}
We consider four possible definitions of the Borda rule. 
Fix any $l\in\mathbb{N}_{+}$. For any $\mathord{\succsim}=(\Sigma_{l},\dots,\Sigma_{1})\in \mathcal{R}_{l}$, let 
\begin{eqnarray}
s^{B1}_{\succsim}(x)&=&\sum_{k\in [l]} \sum_{S\in \Sigma_k: x\in S} \Sigma_k(S)\cdot (k-1),\\
s^{B2}_{\succsim}(x)&=&\sum_{k\in [l]} \sum_{S\in \Sigma_k: x \in S} \Sigma_{k}(S)\cdot \sum_{k'<k} \lvert \Sigma_{k'} \rvert,\\
s^{B3}_{\succsim}(x)&=&\sum_{k\in [l]} \sum_{S\in \Sigma_k: x \in S} \Sigma_{k}(S)\cdot \left (\sum_{k'<k} \lvert \Sigma_{k'} \rvert-\sum_{k'>k} \lvert \Sigma_{k'} \rvert\right),\\
i_{\succsim}(x)&=&\sum_{k\in[l]}\sum_{S\in \Sigma_k: x \in S}\Sigma_{k}(S)\cdot \sum_{k'>k}\sum_{T\in \Sigma_{k'}: x \notin T}\Sigma_{k'}(T).
\end{eqnarray}

\begin{definition}
\label{definition:Borda_rules}
The Borda rule of the first type $R^{B1}$, the Borda rule of the second type $R^{B2}$, the Borda rule of the third type $R^{B3}$, and the Borda rule induced by the inversion number $R^{Bi}$ are defined as follows. For any $\mathord{\succsim}\in\mathcal{R}_{l}$ and $x,y\in X$, 
\begin{itemize}
\item $xR^{B1}_{\succsim}y\Leftrightarrow s^{B1}_{\succsim}(x)\geq s^{B1}_{\succsim}(y)$,
\item $xR^{B2}_{\succsim}y\Leftrightarrow s^{B2}_{\succsim}(x)\geq s^{B2}_{\succsim}(y)$,
\item $xR^{B3}_{\succsim}y\Leftrightarrow s^{B3}_{\succsim}(x)\geq s^{B3}_{\succsim}(y)$,
\item $xR^{Bi}_{\succsim}y\Leftrightarrow i_{\succsim}(x)\leq i_{\succsim}(y)$.
\end{itemize}
\end{definition}

Each of the four definitions evaluate each individual $x$ by their scores ($s^{B1}_{\succsim}(x)$, $s^{B2}_{\succsim}(x)$, $s^{B3}_{\succsim}(x)$, and $i_{\succsim}(x)$). Notice that each of these four scores is defined in such a way that $\sum_{k\in[l]}\sum_{S\in m_{k}:x\in S}\cdots$. This means that each of the four counts the score of $x$ by adding the summands (i.e., the ``$\cdots$'' part) over all coalitions $S$ that include $x$; their differences are the definitions of the summands, which are interpreted as the scores assigned to the \textit{coalitions}. In $s^{B1}$, the score of coalition $S$ is given by the number of equivalence classes ranked lower than $S$; in $s^{B2}$, it is the number of coalitions ranked lower than $S$; in $s^{B3}$, it is the number of coalitions ranked lower than $S$ minus the number of coalitions ranked higher than $S$; and $i$ is defined by applying the equation in Proposition~\ref{proposition:inversion_formula}.

\begin{example}
Let us illustrate them using the coalitional ranking in Example~\ref{example:journals}. We can determine the score as follows: $(s^{B1}_{\succsim}(\text{A}),s^{B1}_{\succsim}(\text{B}),s^{B1}_{\succsim}(\text{C}))=(5,7,3)$, $(s^{B2}_{\succsim}(\text{A}),s^{B2}_{\succsim}(\text{B}),s^{B2}_{\succsim}(\text{C}))=(19,26,12)$, $(s^{B3}_{\succsim}(\text{A}),s^{B3}_{\succsim}(\text{B}),s^{B3}_{\succsim}(\text{C}))=(7,4,-5)$, and $(i_{\succsim}(\text{A}),i_{\succsim}(\text{B}),i_{\succsim}(\text{C}))=(3,3,8)$. Hence, we have $\text{B}P^{B1}_{\succsim}\text{A}P^{B1}_{\succsim}\text{C}$, $BP^{B2}_{\succsim}\text{A}P^{B2}_{\succsim}\text{C}$, $\text{A}P^{B3}_{\succsim}\text{B}P^{B3}_{\succsim}\text{C}$, and $\text{A}I^{Bi}_{\succsim}\text{B}P^{Bi}_{\succsim}\text{C}$. 
\end{example}

In the above example, the social ranking between A and B changes among the four (B is judged better than A in $R^{B1}$ and $R^{B2}$; A is judged better than B in $R^{B3}$; they are indifferent in $R^{Bi}$), but the social ranking between A and C remains the same (A is judged better in all $R^{B1},R^{B2},R^{B3}$, and $R^{Bi}$). In fact, this result arises because $\mathord{\succsim}\notin\mathcal{R}^{\text{A},\text{B}}_{3}$ but $\mathord{\succsim}\in\mathcal{R}^{\text{A},\text{C}}_{3}$. The following theorem formalizes this.

\begin{theorem}
\label{theorem: identity of the four}
Let $x,y\in X$. For any $\mathord{\succsim}=(S_{l},\dots,S_{1})\in \mathcal{L}_{l}^{x,y}$, each of the following is logically equivalent: 
\begin{enumerate}
\item $s^{B1}_{\succsim}(x)\geq s^{B1}_{\succsim}(y)$,
\item $s^{B2}_{\succsim}(x)\geq s^{B2}_{\succsim}(y)$,
\item $s^{B3}_{\succsim}(x)\geq s^{B3}_{\succsim}(y)$,
\item $i_{\succsim}(x)\leq i_{\succsim}(y)$.
\end{enumerate}
\end{theorem}

\begin{proof}
$1\Leftrightarrow 2$ is straightforward. Fix any $\mathord{\succsim}=(S_{l},\dots,S_{1})\in \mathcal{L}_{l}^{x,y}$. Among $S_{l},\dots,S_{1}$, let $S_{i_{p}},\dots,S_{i_{1}}$ ($i_{p}>\dots>i_{1}$) be the sets including $x$, and let $S_{j_{q}},\dots,S_{j_{1}}$ ($j_{q}>\dots>j_{1}$) be the sets including $y$. By the symmetry condition, we have $p=q$. 

By the definitions of $s^{B1}$, $s^{B2}$, $s^{B3}$, $i$, we have: 
\begin{align*}
&s^{B1}_{\succsim}(x)=s^{B2}_{\succsim}(x)=\sum_{t=1}^{p}(i_{t}-1),\\
&s^{B1}_{\succsim}(y)=s^{B2}_{\succsim}(y)=\sum_{t=1}^{p}(j_{t}-1),\\
&s^{B3}_{\succsim}(x)=\sum_{t=1}^{p}\left[(i_{t}-1)-(l-i_{t}-1)\right],\\
&s^{B3}_{\succsim}(y)=\sum_{t=1}^{p}\left[(j_{t}-1)-(l-j_{t}-1)\right],\\
&i_{\succsim}(x)=\sum_{t=1}^{p}\left[(l-i_{t})-(p-t)\right],\\
&i_{\succsim}(y)=\sum_{t=1}^{p}\left[(l-j_{t})-(p-t)\right].
\end{align*}
Hence, 
\begin{align*}
s^{B1}_{\succsim}(x)\geq s^{B1}_{\succsim}(y)&\Leftrightarrow \sum_{t=1}^{p}(i_{t}-1)\geq \sum_{t=~1}^{p}(j_{t}-1)\\
&\Leftrightarrow \sum_{t=1}^{p}i_{t}\geq \sum_{t=1}^{p}j_{t}\\
&\Leftrightarrow \sum_{t=1}^{p}(2i_{t}-l)\geq \sum_{t=1}^{p}(2 j_{t}-l)\\
&\Leftrightarrow s^{B3}_{\succsim}(x)\geq s^{B3}_{\succsim}(y).
\end{align*}

Furthermore, 
\begin{align*}
i_{\succsim}(x)\geq i_{\succsim}(y)&\Leftrightarrow \sum_{t=1}^{p}(-i_{t}+l-p+t)\geq \sum_{t=1}^{p}(-j_{t}+l-p+t)\\
&\Leftrightarrow -\sum_{t=1}^{p}i_{t}\geq -\sum_{t=1}^{p}j_{t}\\
&\Leftrightarrow -s^{B1}_{\succsim}(x)\geq -s^{B1}_{\succsim}(y)\\
&\Leftrightarrow s^{B1}_{\succsim}(x)\leq s^{B1}_{\succsim}(y).
\end{align*}

\end{proof}

Thus, if we consider a coalition order $\succsim$ in $\mathcal{L}^{x,y}$ (i.e., a linear order such that the same number of coalitions include $x$ and $y$, respectively), then there is no difference between the possible definitions $R^{B1}$, $R^{B2}$, $R^{B3}$, and $R^{Bi}$. This fact leads us to ask which of these definitions are well performing in the general domain of $\mathcal{R}$, which is the main theme of the next section. 

\begin{proposition}
\label{proposition:Borda_rules_and_ECON}
Among $R^{B1}$, $R^{B2}$, $R^{B3}$, and $R^{Bi}$, only $R^{B1}$ satisfies ECON.
\end{proposition}

\begin{proof}
The fact that $R^{B1}$ satisfies ECON is straightforward. We will prove that none of $R^{B2}$, $R^{B3}$, and $R^{Bi}$ satisfy ECON by presenting counterexamples. 

For $R^{B2}$, let $\mathord{\succsim}=(\mset{\{x\}},\mset{\{y\}},\mset{\{y\}},\mset{\{z\}})$ and $\mathord{\succsim}'=(\mset{\{z\}},\mset{\{z\}},\mset{\{z\}},\mset{\{z\},\{z\}})$.
Then, we can verify that $xI^{B2}_{\succsim}y$, $xI^{B2}_{\succsim'}y$, and $yP^{B2}_{\succsim+\succsim'}x$. This contradicts ECON.

For $R^{B3}$, let $\mathord{\succsim}=(\mset{\{x\}},\mset{\{y\},\{y\},\{y\}},\mset{\{z\},\{z\},\{z\}})$ and $\mathord{\succsim}'=(\mset{\{z\}},\mset{\{z\}},\mset{\{z\}})$. Then, we have $xI^{B3}_{\succsim}y$, $xI^{B3}_{\succsim'}y$, and $xP^{B3}_{\succsim+\succsim'}y$. This contradicts ECON. 

For $R^{Bi}$, let $\mathord{\succsim}=(\mset{\{x\}},\mset{\{y\},\{y\}},\mset{\{x\}})$ and $\mathord{\succsim}'=(\mset{\{z\},\{z\}},\mset{\{z\}},\mset{\{z\}})$. Then, we have $xI^{Bi}_{\succsim}y$, $xI^{Bi}_{\succsim}y$, and $xP^{Bi}_{\succsim}y$. This contradicts ECON. 
\end{proof}

Finally, we consider the relationship between our Borda rules and the Banzhaf value. For any coalitional ranking $\mathord{\succsim} \in \mathcal{R}_l$, define the following values for coalitions $S \in \mathcal{X}$:
\begin{eqnarray}
v^{B1}_{\succsim}(S)&=&\sum_{k\in [l]}(k-1)\cdot  \Sigma_k(S),\\
v^{B2}_{\succsim}(S)&=&\sum_{k\in [l]}  \Sigma_{k}(S)\cdot \sum_{k'<k} \lvert \Sigma_{k'} \rvert,\\
v^{B3}_{\succsim}(S)&=&\sum_{k\in [l]}  \Sigma_{k}(S)\cdot \left (\sum_{k'<k} \lvert \Sigma_{k'} \rvert-\sum_{k'>k} \lvert \Sigma_{k'} \rvert\right).
\end{eqnarray}
Notice that for all $x\in X$,
\begin{eqnarray}
s^{B1}_{\succsim}(x)&=&\sum_{k\in [l]}(k-1)\cdot \sum_{S\in \Sigma_k: x\in S} \Sigma_k(S)\nonumber \\
&=&\sum_{S\in \mathcal{X}: x\in S} \sum_{k\in [l]}(k-1)\cdot  \Sigma_k(S)\nonumber \\
&=&\sum_{S\in \mathcal{X}: x\in S} v^{B1}_\succsim(S).\nonumber
\end{eqnarray}
Similarly, one can check that for all $x\in X$,
$$s^{B2}_{\succsim}(x)=\sum_{S\in \mathcal{X}: x\in S} v^{B2}_\succsim(S) \ \ \mbox{ and } \ \ s^{B3}_{\succsim}(x)=\sum_{S\in \mathcal{X}: x\in S} v^{B3}_\succsim(S).$$\\
Moreover, the pair $(X,v^{Bt}_{\succsim})$, for each $t\in \{1,2,3\}$, can be interpreted as an $|X|$-person game in characteristic function form \citep{owen2013game}. Therefore, a classical power index, the Banzhaf value \citep{banzhaf1964weighted}, can be computed on the game $(X,v^{Bt}_{\succsim})$ for any $t\in \{1,2,3\}$ and $x \in X$, as follows:
\[
\beta_x(v^{Bt}_{\succsim})=\sum_{T\subseteq X\backslash\{x\}}\frac{1}{2^{|X|-1}}(v^{Bt}_{\succsim}(T\cup\{i\})-v^{Bt}_{\succsim}(T)).
\]
Then, we can use the Banzhaf value to define an SRS $R^{\beta t}$ for any $t\in \{1,2,3\}$ as follows. For any $\mathord{\succsim} \in \mathcal{R}_l$ and $x,y\in X$, define $R^{\beta t}$ such that
\[xR^{\beta t}_{\succsim}y\iff \beta_x(v^{Bt}_{\succsim})\geq \beta_y(v^{Bt}_{\succsim}).
\]
 In the following proposition,  we show that each ranking provided by the Banzhaf value on each characteristic function $v^{Bt}$, for $t\in \{1,2,3\}$,  equals the ranking provided by the corresponding  solution $R^{Bt}$.
\begin{proposition}
\label{proposition:banzhaf}
For any $\mathord{\succsim}=(S_{l},\dots,S_{1})\in \mathcal{R}_{l}$,
it holds that 
\[
R^{\beta t}=R^{Bt}
\]
for any $t\in \{1,2,3\}$.
\end{proposition}
\begin{proof}
For any $x,y \in X$ and $t\in \{1,2,3\}$, we have
\begin{eqnarray}
\label{i vs j}
&&\beta_x(v^{Bt}_{\succsim})-\beta_y(v^{Bt}_{\succsim})\nonumber \\
&=&\frac{1}{2^{|X|-1}}\sum_{T\subseteq X\backslash\{x\}}(v^{Bt}_{\succsim}(T\cup\{x\})
-v^{Bt}_{\succsim}(T))\nonumber \\
&-&\frac{1}{2^{|X|-1}}\sum_{T\subseteq X\backslash\{y\}}(v^{Bt}_{\succsim}(T\cup\{y\})-v^{Bt}_{\succsim}(T))\nonumber \\
&=&\frac{1}{2^{|X|-1}}\sum_{T\subseteq X\backslash\{x,y\}}\Big(\big((v^{Bt}_{\succsim}(T\cup\{x\})-v^{Bt}(T))+(v^{Bt}_{\succsim}(T\cup\{x,y\})-v^{Bt}(T\cup\{y\}))\big)\nonumber \\
&-&\big( (v^{Bt}_{\succsim}(T\cup\{y\})-v^{Bt}_{\succsim}(T))+(v^{Bt}_{\succsim}(T\cup\{x,y\})-v^{Bt}_{\succsim}(T\cup\{x\}))\big) \Big)\nonumber \\
&=&\frac{1}{2^{|X|-2}}\sum_{T\subseteq X\backslash\{x,y\}}(v^{Bt}_{\succsim}(T\cup\{x\})-v^{Bt}_{\succsim}(T\cup\{y\}))\nonumber \\
&=&\frac{1}{2^{|X|-2}}(s^{Bt}_{\succsim}(x)-s^{Bt}_{\succsim}(y)),\nonumber
\end{eqnarray}
where the last equality follows directly from the fact that for any $x \in X$ and $t\in \{1,2,3\}$,
\begin{eqnarray}
s^{Bt}_{\succsim}(x)&=&\sum_{S\in \mathcal{X}: x\in S}v^{Bt}_{\succsim}(S).\nonumber
\end{eqnarray}
\end{proof}
It is well-known that the Banzhaf value $\beta$ satisfies the property of strong monotonicity, which states that $\beta_x(v)\geq\beta_x(w)$ for all games $(X,v)$ and $(X,w)$ and elements $x \in X$ such that $v(S\cup \{x\})-v(S)\geq w(S\cup \{x\})-v(S)$
for all $S\in \mathcal{X}$ with $x \notin S$. So, in view of Theorem~\ref{theorem: identity of the four} specifying that considering a coalition order $\succsim$ in $\mathcal{L}^{x,y}$, solutions $R^{B1}$, $R^{B2}$, and $R^{B3}$ provide the same ranking, the Banzhaf value on games $v^{Bt}_{\succsim}$, for $t\in \{1,2,3\}$, could be a good candidate index to assess the relative intensities of the position of each element in the common ranking according to the alternative scores $B1$, $B2$, and $B3$. However, this matter is beyond the scope of this study.

\section{Main results}
\label{section:main result}

\subsection{Characterization}

As discussed in the previous section, for coalitional rankings that are linear and where the number of coalitions including two elements $x$ and $y$ is exactly the same, it seems reasonable to rank $x$ and $y$ according to the inversion number (the lower the inversion number, the higher the rank).

The results in the previous section established a close relationship between the closeness to unanimity, one of the fundamental properties of the Borda rule, and the inversion number.  We then introduced four possible definitions of the Borda rule (Definition~\ref{definition:Borda_rules}), each of which is equivalent to the inversion number on linear and symmetric domains (Theorem~\ref{theorem: identity of the four}). However, we also proved that only the first one, $R^{B1}$, satisfies ECON (Proposition~\ref{proposition:Borda_rules_and_ECON}). The purpose of this section is to specify a reasonable definition of the Borda SRS through the two fundamental conditions---the closeness to unanimity and ECON---together with complementary axioms. Here, we state a formal definition regarding the former. 

\begin{definition}
We say that an SRS satisfies \textit{closeness to unanimity} (CU) if, for any $l\in\mathbb{N}_{+}$ $x,y\in X$, and $\mathord{\succsim}\in\mathcal{L}^{x,y}_{l}$, we have $R_{\succsim}\mid_{\{x,y\}}=R^{Bi}_{\succsim}\mid_{\{x,y\}}$. 
\end{definition}

CU demands that the SRS coincides with the inversion number under linear and symmetric domains and hence, it must be the ``closest to unanimity'' as shown in Section~\ref{subsection:Borda rule and the inversion number}. Notice that Theorem~\ref{theorem: identity of the four} implies that $R^{B1}$, $R^{B2}$, $R^{B3}$, and $R^{Bi}$ all satisfy CU. Thus, the purpose of this section can be rephrased as finding a reasonable definition of the Borda SRS that satisfies CU and ECON. Our next theorem answers this question.

\begin{theorem}
\label{theorem: characterization of B1}
An SRS $R$ satisfies ECON, DCONT, NT, IPP, and CU if and only if $R=R^{B1}$. 
\end{theorem}

\begin{lemma}
\label{lemma: pair cancellation}
Let $R$ be an SRS that satisfies ECON, DCONT, NT, and CU. For any $l\in\mathbb{N}_{+}$, $x,y\in X$, $\mathord{\succsim}=(\Sigma_{l},\dots,\Sigma_{1})\in\mathcal{R}_{l}$, $S,T\in \mathcal{X}$ such that $x\in S\not\ni y$ and $y\in T\not\ni x$, and $k_{1}\in[l]$, we have $R_{\succsim+(1_{S})_{k_{1}}+(1_{T})_{k_{1}}}\mid_{\{x,y\}}=R_{\succsim}\mid_{\{x,y\}}$.
\end{lemma}

\begin{proof}
The proof is made up of two parts. 

Step~1. We will prove that $xI_{\succsim_{1}}y$, where $\succsim_{1}:=\left(\sum_{k\neq k_{1}}(1_{X})_{k}\right)+(1_{S})_{k_{1}}+(1_{T})_{k_{1}}$. 
Suppose to the contrary that $yP_{\succsim_{1}}x$. 
Let 
\[\mathord{\succsim}_{2}:=\mathord{\succsim}_{1}+\sum_{k\in[l]}(1_{X})_{k}.\]
Since $xI_{\sum_{k\in[l]}(1_{X})_{k}}y$ holds by NT, ECON implies that $R_{\succsim_{2}}\mid_{\{x,y\}}=R_{\succsim_{1}}\mid_{\{x,y\}}$. Hence, we obtain $yP_{\succsim_{2}}x$. 

Now, let
\[\mathord{\succsim}_{3}:=\left(\sum_{k\neq k_{1},k_{1}-1}(1_{X})_{k}\right)+(1_{S})_{k_{1}}+(1_{T})_{k_{1}-1}.\]
Note that $\mathord{\succsim}_{3}\in\mathcal{L}_{l}^{x,y}$. So, CU implies that $xP_{\succsim_{3}}y$. 
Let 
\[\mathord{\succsim}_{4}:=\mathord{\succsim}_{3}+\left(\sum_{k\in[l]}(1_{X})_{k}\right)+(1_{X})_{k_{1}-1}.\]
Since $xI_{\left(\sum_{k\in[l]}(1_{X})_{k}\right)+(1_{X})_{k_{1}-1}}y$ holds by NT, ECON implies that $R_{\succsim_{4}}\mid_{\{x,y\}}=R_{\succsim_{3}}\mid_{\{x,y\}}$. Therefore, we obtain $xP_{\succsim_{4}}y$. 

Notice that $\succsim_{2}$ is obtained from $\succsim_{4}$ by upgrading $\Gamma:=1_{T}$ by one class. Since $xP_{\succsim_{4}}y$, we have $\neg(yP_{\succsim_{2}}x)$ by DCONT. This contradicts $yP_{\succsim_{2}}x$. So, it follows that $\neg (yP_{\succsim_{1}}x)$. By the completeness of $R_{\succsim_{1}}$, we can say that $xR_{\succsim_{1}}y$. By swapping $x$ and $y$ in the arguments above, we can also say that $yR_{\succsim_{1}}x$. Hence, we can conclude that $xI_{\succsim_{1}}y$. 

Step~2. Proof of the lemma. Let $\mathord{\succsim}_{5}:=\left(\sum_{k\neq k_{1}}(1_{X})_{k}\right)+\left(\sum_{k\in[l]}(1_{X})_{k}\right)$. Then, NT implies that $xI_{\succsim_{5}}y$. Hence, ECON says that 
\[R_{\succsim+(1_{S})_{k_{1}}+(1_{T})_{k_{1}}}\mid_{\{x,y\}}=R_{\succsim+(1_{S})_{k_{1}}+(1_{T})_{k_{1}}+\succsim_{5}}\mid_{\{x,y\}}.\]
Since $xI_{\left(\sum_{k\neq k_{1}}(1_{X})_{k}\right)+(1_{S})_{k_{1}}+(1_{T})_{k_{1}}}y$ by Step~1, ECON implies that 
\[R_{\succsim+(1_{S})_{k_{1}}+(1_{T})_{k_{1}}+\succsim_{5}}\mid_{\{x,y\}}=R_{\succsim+\sum_{k\in[l]}(1_{X})_{k}}\mid_{\{x,y\}}.\]
Since $xI_{\sum_{k\in[l]}(1_{X})_{k}}y$ by NT, the right-hand side is equal to $R_{\succsim}\mid_{\{x,y\}}$ (by ECON). To summarize, we have shown that $R_{\succsim+(1_{S})_{k_{1}}+(1_{T})_{k_{1}}}\mid_{\{x,y\}}=R_{\succsim}\mid_{\{x,y\}}$. 
\end{proof}

\begin{definition}
An SRS $R$ satisfies \textit{cancellation} (CAN) if, for any $\mathord{\succsim}=(\Sigma_{l},\dots,\Sigma_{1})\in\mathcal{R}_{l}$, $x,y\in X$, $k_{1},k_{2}\in[l]$, and $S,T\in \mathcal{X}$ such that $x\in S\not\ni y$ and $y\in T\not\ni x$, we have that 
$R_{\succsim+(1_{S})_{k_{1}}+(1_{T})_{k_{2}}}=R_{\succsim+(1_{S})_{k_{1}-1}+(1_{T})_{k_{2}-1}}$.
\end{definition}

\begin{proposition}
If an SRS $R$ satisfies ECON, DCONT, NT, and CU, then $R$ also satisfies CAN.  
\end{proposition}

\begin{proof}
Let $l\in\mathbb{N}_{+}$. Let $\mathord{\succsim}=(\Sigma_{l},\dots,\Sigma_{1})\in\mathcal{R}_{l}$, $x,y\in X$, and $S,T\in \mathcal{X}$ such that $x\in S\not\ni y$ and $y\in T\not\ni x$, $k_{1},k_{2}\in[l]\setminus\{1\}$. We will prove that
\[R_{\succsim+(1_{S})_{k_{1}}+(1_{T})_{k_{2}}}\mid_{\{x,y\}}=R_{\succsim+(1_{S})_{k_{1}-1}+(1_{T})_{k_{2}-1}}\mid_{\{x,y\}}.\]
Without loss of generality, we can assume that $k_{1}\geq k_{2}$. We consider three cases. 

Case~1: $k_{1}=k_{2}$. 
By applying Lemma~\ref{lemma: pair cancellation}, we have that 
$R_{\succsim+(1_{S})_{k_{1}}+(1_{T})_{k_{1}}}\mid_{\{x,y\}}=R_{\succsim}\mid_{\{x,y\}}$ and $R_{\succsim+(1_{S})_{k_{1}-1}+(1_{T})_{k_{1}-1}}\mid_{\{x,y\}}=R_{\succsim}\mid_{\{x,y\}}$. Hence, the lemma follows. 

Case~2: $k_{2}=k_{1}-1$. 
The proof of this case is made up of two steps. 

Step~1: We first prove that $xI_{\succsim_{1}}y$, where 
\[\mathord{\succsim}_{1}:=\left(\sum_{k\neq k_{1},k_{1}-1,k_{1}-2}(1_{X})_{k}\right)+(1_{T})_{k_{1}}+(1_{S})_{k_{1}-1}+(1_{S})_{k_{1}-1}+(1_{T})_{k_{1}-2}.\]
Let $\mathord{\succsim}_{2}:=\left(\sum_{k\neq k_{1},k_{1}-1,k_{1}-2}(1_{X})_{k}\right)+(1_{T})_{k_{1}}+(1_{S})_{k_{1}-1}+(1_{S})_{k_{1}-2}+(1_{T})_{k_{1}-2}$, and let 
$\mathord{\succsim}_{3}:=\mathord{\succsim}_{2}+\sum_{k\in[l]}(1_{X})_{k}$. Since $xI_{\sum_{k\in[l]}(1_{X})_{k}}y$ by NT, ECON demands that 
\[R_{\succsim_{2}}\mid_{\{x,y\}}=R_{\succsim_{3}}\mid_{\{x,y\}}.\]
Notice that $\mathord{\succsim}_{3}=\mathord{\succsim}_{3}'+\mathord{\succsim}_{3}''$, where $\mathord{\succsim}_{3}':=\left(\sum_{k\neq k_{1},k_{1}-1}(1_{X})_{k}\right)+(1_{T})_{k_{1}}+(1_{S})_{k_{1}-1}$ and $\mathord{\succsim}_{3}'':=\left(\sum_{k\neq k_{1}-2}(1_{X})_{k}\right)+(1_{S})_{k_{1}-2}+(1_{T})_{k_{1}-2}$. By Step~1 of the proof of Lemma~\ref{lemma: pair cancellation}, we have $xI_{\succsim''_{3}}y$. Hence, ECON says that $R_{\succsim_{3}}\mid_{\{x,y\}}=R_{\succsim'_{3}}\mid_{\{x,y\}}$. By CU, we have $yP_{\succsim'_{3}}x$. Hence, we obtain that $yP_{\succsim_{3}}x$, and so $yP_{\succsim_{2}}x$. Since $\succsim_{1}$ is obtained from $\succsim_{2}$ by upgrading the position of $\Gamma:=1_{S}$ from the $(k_{1}-2)$-th class to the $(k_{1}-1)$-th class, DCONT implies $yR_{\succsim_{1}}x$. 

Let $\mathord{\succsim}_{4}:=\left(\sum_{k\neq k_{1},k_{1}-1,k_{1}-2}(1_{X})_{k}\right)+(1_{T})_{k_{1}}+(1_{S})_{k_{1}}+(1_{S})_{k_{1}-1}+(1_{T})_{k_{1}-2}$, and let $\mathord{\succsim}_{5}:=\mathord{\succsim}_{4}+\sum_{k\in[l]}(1_{X})_{k}$. By NT and ECON, we have that $R_{\succsim_{4}}\mid_{\{x,y\}}=R_{\succsim_{5}}\mid_{\{x,y\}}$. Notice that $\mathord{\succsim}_{5}=\mathord{\succsim}_{5}'+\mathord{\succsim}_{5}''$, where $\mathord{\succsim}_{5}':=\left(\sum_{k\neq k_{1}-1,k_{1}-2}(1_{X})_{k}\right)+(1_{S})_{k_{1}-1}+(1_{T})_{k_{1}-2}$ and $\mathord{\succsim}''_{5}:=\left(\sum_{k\neq k_{1}}(1_{X})_{k}\right)+(1_{S})_{k_{1}}+(1_{T})_{k_{1}}$. By Step~1 of the proof of Lemma~\ref{lemma: pair cancellation}, we have that $xI_{\succsim''_{5}}y$. Hence, ECON demands that $R_{\succsim_{5}}\mid_{\{x,y\}}=R_{\succsim_{5}'}\mid_{\{x,y\}}$. By CU, we have $xP_{\succsim_{5}'}y$. Hence, we can conclude that $xP_{\succsim_{4}}y$. Since $\mathord{\succsim}_{1}$ is obtained from $\mathord{\succsim}_{4}$ by degrading $\Gamma=1_{S}$ from the $k_{1}$-th class to the $(k_{1}-1)$-th class, DCONT demands that $xR_{\succsim_{1}}y$. 
We have shown that $yR_{\succsim_{1}}x$ and $xR_{\succsim_{1}}y$. These imply $xI_{\succsim_{1}}y$.

Step~2: We prove the lemma. Let $\mathord{\succsim}_{6}:=\mathord{\succsim}+(1_{S})_{k_{1}}+(1_{T})_{k_{1}-1}+\mathord{\succsim}_{1}+\sum_{k\in[l]}(1_{X})_{k}$. Because $xI_{\succsim_{1}}y$ (Step~1) and $xI_{\sum_{k\in[l]}(1_{X})_{k}}y$ by NT, we have that 
\[R_{\succsim+(1_{S})_{k_{1}}+(1_{T})_{k_{1}-1}}\mid_{\{x,y\}}=R_{\succsim_{6}}\mid_{\{x,y\}}.\]
Notice that $\mathord{\succsim}_{6}=\mathord{\succsim}+(1_{S})_{k_{1}-1}+(1_{T})_{k_{1}-2}+\mathord{\succsim}_{6}'+(1_{S}+1_{T})_{k_{1}}+(1_{S}+1_{T})_{k_{1}-1}$, where $\mathord{\succsim}'_{6}:=\sum_{k\neq k_{1},k_{1}-1,k_{1}-2}(1_{X})_{k}+\sum_{k\in[l]}(1_{X})_{k}$. By Lemma~\ref{lemma: pair cancellation}, adding the last two terms, $(1_{S}+1_{T})_{k_{1}}$ and $(1_{S}+1_{T})_{k_{1}-1}$, does not affect the social ranking between $x$ and $y$. Furthermore, $xI_{\succsim'_{6}}y$ holds by NT. Hence, ECON says that $R_{\succsim_{6}}\mid_{\{x,y\}}=R_{\succsim+(1_{S})_{k_{1}-1}+(1_{T})_{k_{1}-2}}\mid_{\{x,y\}}$. To conclude, we obtain that 
\[R_{\succsim+(1_{S})_{k_{1}}+(1_{T})_{k_{1}-1}}\mid_{\{x,y\}}=R_{\succsim+(1_{S})_{k_{1}-1}+(1_{T})_{k_{1}-2}}\mid_{\{x,y\}}.\]

Case~3: $k_{2}\leq k_{1}-2$. 
Now, let $\mathord{\succsim}_{7}:=\mathord{\succsim}+(1_{S})_{k_{1}}+(1_{T})_{k_{2}}+\mathord{\succsim}_{8}+\sum_{k\in[l]}(1_{X})_{k}$, where $\mathord{\succsim}_{8}:=\left(\sum_{k\neq k_{1},k_{1}-1,k_{2},k_{2}-1}(1_{X})_{k}\right)+(1_{T})_{k_{1}}+(1_{S})_{k_{1}-1}+(1_{S})_{k_{2}}+(1_{T})_{k_{2}-1}$. By CU, we have $xI_{\succsim_{8}}y$. By NT, we have $\sum_{k\in[l]}(1_{X})_{k}$. Hence, we have $R_{\succsim_{7}}\mid_{\{x,y\}}=R_{\succsim+(1_{S})_{k_{1}}+(1_{T})_{k_{2}}}\mid_{\{x,y\}}$. Notice that $\mathord{\succsim}_{7}=\mathord{\succsim}+(1_{S})_{k_{1}-1}+(1_{T})_{k_{2}-1}+\mathord{\succsim}'_{7}+\mathord{\succsim}''_{7}$, where $\mathord{\succsim}'_{7}=(1_{S}+1_{T})_{k_{1}}+(1_{S}+1_{T})_{k_{2}}$ and $\mathord{\succsim}''_{7}:=\left(\sum_{k\neq k_{1},k_{1}-1,k_{2},k_{2}-1}(1_{X})_{k}\right)+\sum_{k\in[l]}(1_{X})_{k}$. We have $xI_{\succsim''_{7}}y$ by NT. Furthermore, Lemma~\ref{lemma: pair cancellation} says that we can drop $\succsim_{7}'$ without changing the social ranking between $x$ and $y$. Hence, we have 
$R_{\succsim_{7}}\mid_{\{x,y\}}=R_{\succsim+(1_{S})_{k_{1}-1}+(1_{T})_{k_{1}-2}}$. This completes the proof. 
\end{proof}

\begin{definition}
An SRS $R$ satisfies \textit{monotonicity} (MON) if, for any $\mathord{\succsim}=(\Sigma_{l},\dots,\Sigma_{1})\in\mathcal{R}_{l}$, $x,y\in X$, $\hat{k}\in[l]$, and $S\in \mathcal{X}$ such that $x\in S \not\ni y$, we have that $xR_{\succsim+(1_{S})_{\hat{k}}}y\Rightarrow xP_{\succsim+(1_{S})_{\hat{k}+1}}y$.
\end{definition}

\begin{proposition}
If an SRS $R$ satisfies ECON, DCONT, NT, and CU, then $R$ also satisfies MON.  
\end{proposition}

\begin{proof}
Let $l\in\mathbb{N}_{+}$, $k_{1}\in[l]$, $\mathord{\succsim}=(\Sigma_{l},\dots,\Sigma_{1})\in\mathcal{R}_{l}$, $x,y\in X$, and $S\in\mathcal{X}$ with $x\in S\not\ni y$. Assume that $xR_{\succsim+(1_{S})_{k_{1}}}y$. Let $T:=\{y\}$.

Let $\mathord{\succsim}_{1}:=\left(\sum_{k\neq k_{1}+1,k_{1}}(1_{X})_{k}\right)+(1_{S})_{k_{1}+1}+(1_{T})_{k_{1}}$. By CU, we have $xP_{\succsim_{1}}y$. Let $\mathord{\succsim}_{2}:=\mathord{\succsim}_{1}+\sum_{k\in[l]}(1_{X})_{k}$. Since $xI_{\sum_{k\in[l]}(1_{X})k}y$ by NT, we have that $R_{\succsim_{1}}\mid_{\{x,y\}}=R_{\succsim_{2}}\mid_{\{x,y\}}$ by ECON. Hence, $xP_{\succsim_{2}}y$. 

By ECON, we have $xP_{\succsim+(1_{S})_{k_{1}}+\succsim_{2}}y$. Notice that $\mathord{\succsim}+(1_{S})_{k_{1}}+\mathord{\succsim}_{2}=\mathord{\succsim}+(1_{S})_{k_{1}+1}+(1_{S}+1_{T})_{k_{1}}+\sum_{k\neq k_{1}+1,k_{1}}(1_{X})_{k}+\sum_{k\in[l]}(1_{X})_{k}$. Since $xI_{\sum_{k\neq k_{1}+1,k_{1}}(1_{X})_{k}}y$ and $xI_{\sum_{k\in[l]}(1_{X})_{k}}y$ by NT, we have that 
$R_{\succsim+(1_{S})_{k_{1}}+\succsim_{2}}\mid_{\{x,y\}}=R_{\succsim+(1_{S})_{k_{1}+1}+(1_{S}+1_{T})_{k_{1}}}\mid_{\{x,y\}}$. By Lemma~\ref{lemma: pair cancellation}, the right-hand side is equal to $R_{\succsim+(1_{S})_{k_{1}+1}}$. These show that $xP_{\succsim+(1_{S})_{k_{1}+1}}y$. 
\end{proof}

\begin{proof}[Proof of Theorem~\ref{theorem: characterization of B1}]
For simplicity, we will denote $s^{B1}$ as $s$. For $\mathord{\succsim}\in\mathcal{R}$ and $x\in X$, we call $s
_{\succsim}(x)$ the (Borda) score of $x$ (at $\succsim$).
Let $\mathord{\succsim}=(\Sigma_{l},\dots,\Sigma_{1})\in\mathcal{R}_{l}$ and $x,y\in X$. 

Let $\mathord{\succsim}':=\mathord{\succsim}+(1_{\{x\}}+1_{\{y\}})_{1}+(1_{\{x\}}+1_{\{y\}})_{1}+\dots+(1_{\{x\}}+1_{\{y\}})_{1}$, where $(1_{\{x\}}+1_{\{y\}})_{1}$ appears $l$ times. By IPP, we have $R_{\succsim'}\mid_{\{x,y\}}=R_{\succsim}\mid_{\{x,y\}}$. 

Let $\mathord{\succsim}'':=\mathord{\succsim}+\sum_{k\in[l]}(1_{\{x\}}+1_{\{y\}})_{k}$; that is, $\succsim''$ is obtained from $\succsim'$ by shifting $(l-1)$ $(1_{\{x\}}+1_{\{y\}})_{1}$'s to the upper classes. By CAN, we have $R_{\succsim'}\mid_{\{x,y\}}=R_{\succsim''}\mid_{\{x,y\}}$. 
We will prove that (i)~if $s_{\succsim}(x)=s_{\succsim}(y)$, then $xI_{\succsim}y$, and (ii)~if $s_{\succsim}(x)>s_{\succsim}(y)$, then $xP_{\succsim}y$. 

Proof of (i): Assume $s_{\succsim}(x)=s_{\succsim}(y)$. If there exists $S\in\Sigma_{l}+\dots+\Sigma_{2}$ such that $x\in S\not\ni y$, then it follows that $T\in \Sigma_{l}+\dots+\Sigma_{2}$ such that $y\in T\not\ni x$ (otherwise, $s_{\succsim}(x)$ must be greater than $s_{\succsim}(y)$). By CAN, we can downgrade them by one class without changing the social ranking between $x$ and $y$. Notice that this process decreases the scores of each $x$ and $y$ by one unit. If there is still another $S'$ in the second or higher classes such that $x\in S'\not\ni y$, then there should also exist $T'$ in the second or higher classes such that $y\in T'\not\ni x$ (because the Borda scores of $x$ and $y$ must be the same). By CAN, we can downgrade them by one class respectively without changing the social ranking between $x$ and $y$. We can repeat this simultaneous downgrading process until all the $(x,y)$-discriminating coalitions in $\Sigma_{l}+\dots+\Sigma_{2}$ are relegated into the worst class. Let $\succsim'$ be the resulting coalitional ranking. If there are no $(x,y)$-discriminating coalitions in $\Sigma_{l}+\dots+\Sigma_{2}$, then $\mathord{\succsim}' = \mathord{\succsim}$. By the definition of $\succsim'$, we have $R_{\succsim}\mid_{\{x,y\}}=R_{\succsim'}\mid_{\{x,y\}}$. 

Let $\succsim''$ be a coalitional ranking obtained from $\succsim'$ by dropping all $(x,y)$-discriminating sets in the first class except one pair of $S$ and $T$ (notice that such $S$ and $T$ exist by the definition of $\succsim'$). By IPP, we have $R_{\succsim'}\mid_{\{x,y\}}=R_{\succsim''}\mid_{\{x,y\}}$. Furthermore, note that all $l$ classes are symmetric with respect to $x$ and $y$. Hence, NT demands that $xI_{\succsim''}y$. Thus, we can conclude that $xI_{\succsim}y$.

Proof of (ii): Assume $s_{\succsim}(x)>s_{\succsim}(y)$. Then, we can take some $S\in\Sigma_{l}+\dots+\Sigma_{2}$ with $x\in S\not\ni y$ and downgrade it by one class. If the Borda score of $x$ is still higher than that of $y$, we can take some $S'$ with $x\in S'\not\ni y$ and downgrade it again. Repeat this process until the Borda scores of $x$ and $y$ are the same. Let $\succsim'$ be the resulting coalitional ranking. By (i), we have $xI_{\succsim'}y$. Since $\succsim$ is obtained from $\succsim'$ by upgrading some coalitions that include $x$ and not $y$, MON says that $xP_{\succsim}y$.
\end{proof}

Recall that we proved Theorem~\ref{theorem: characterization of B1} by first showing some logical consequences (CAN and MON) of the four axioms (ECON, DCONT, NT, and CU) and then using only CAN, MON, NT, and IPP. Hence, our proof also shows the following. 

\begin{corollary}
\label{corollary:alternative characterization of B1}
An SRS $R$ satisfies MON, CAN, NT, and IPP if and only if $R=R^{B1}$.
\end{corollary}

\subsection{Independence of the axioms}

We define an SRS called \textit{non-worst lexcel}, denoted $R^\text{NWL}$, as follows. For any $\mathord{\succsim}=(\Sigma_{l},\dots,\Sigma_{1})\in\mathcal{R}_{l}$ and $x,y\in X$, we have $xR^\text{NWL}_{\succsim}y\Leftrightarrow (x_{l},\dots,x_{2})\geq^{L}(y_{l},\dots,y_{2})$, where for $z\in\{x,y\}$ and $k\in[l]$, we let $z_{k}:=\sum_{S\ni z}\Sigma_{k}(S)$. In other words, $R^\text{NWL}$ is a variant of lexcel that disregards the worst class. 
Let $R^\text{const}$ be a constant SRS such that for any $\mathord{\succsim}\in\mathcal{R}$, $R_{\succsim}^{I}=X\times X$. 

We now introduce three variants of $R^{B1}$. For $l\in \mathbb{N}_{+}$, $\mathord{\succsim}\in\mathcal{R}_{l}$, and $x,y\in X$, we say that $\{x,y\}$ is a \textit{buddy pair} at $\succsim$ if $\sum_{k\in[l]}\sum_{S\supseteq \{x,y\}}\Sigma_{k}(S)> l$ (i.e., the number of coalitions that $x$ and $y$ join simultaneously is greater than the number of equivalence classes; in this sense, being a buddy pair means frequent collaboration). 

Let $\vartriangleright$ be a fixed linear order on $X$. Let $R^{B1,\vartriangleright}$ be an SRS obtained from $R^{B1}$ by breaking the ties according to $\vartriangleright$. That is, (i)~if $xI^{B1}_{\succsim}y$, then $R^{B1,\vartriangleright}_{\succsim}\mid_{\{x,y\}}=\vartriangleright\mid_{\{x,y\}}$, and (ii)~otherwise, $R^{B1,\vartriangleright}_{\succsim}\mid_{\{x,y\}}=R^{B1}_{\succsim}\mid_{\{x,y\}}$. 

Let $R^{B1:\vartriangleright}$ be defined as follows: for any $x,y\in X$ with $x\vartriangleright y$ and for any $\mathord{\succsim}=(\Sigma_{l},\dots,\Sigma_{1})\in \mathcal{R}_{l}$, we have (i)~if $l=2$ and there exists $S\in \mathcal{X}$ such that $x\in S\not\ni y$ and $\Sigma_{2}(S)>0$, then $xP^{B1:\vartriangleright}_{\succsim}y$, (ii)~if $l=2$ and there does not exist such an $S$, then $xI_{\succsim}y$, and (iii)~if $l\neq 2$, then $R^{B1:\vartriangleright}\mid_{\{x,y\}}=R^{B1}_{\succsim}\mid_{\{x,y\}}$. 

We define $R^{B1:\text{const}}$ as follows: 
(i)~if $\{x,y\}$ is a buddy pair, then $R^{B1:\text{const}}_{\succsim}\mid_{\{x,y\}}=R^\text{const}_{\succsim}\mid_{\{x,y\}}$, and (ii)~if $\{x,y\}$ is not a buddy pair, then $R^{B1:\text{const}}_{\succsim}\mid_{\{x,y\}}=R^{B1}_{\succsim}\mid_{\{x,y\}}$. 

Finally, $R^{B1:\text{NWL}}$ is defined as follows: (i)~if $l= 3$, then $R^{B1:\text{NWL}}_{\succsim}\mid_{\{x,y\}}=R^\text{NWL}_{\succsim}\mid_{\{x,y\}}$, and (ii)~if $l\neq 3$, then $R^{B1:\text{NWL}}_{\succsim}\mid_{\{x,y\}}=R^{B1}_{\succsim}\mid_{\{x,y\}}$.

Let $R^{\tilde{B1}}$ be an SRS such that $xR^{\tilde {B1}}_{\succsim}y$ if and only if $\sum_{k\in[l]}\sum_{S\ni x}k\cdot \Sigma_{k}(S)\geq \sum_{k\in[l]}\sum_{T\ni y}k\cdot \Sigma_{k}(T)$. The score assigned to the $k$-th-worst class is replaced with $k$ (for $R^{B1}$, it was $k-1$).

\begin{proposition}
The five axioms in Theorem~\ref{theorem: characterization of B1} are logically independent. In fact, 

\begin{itemize}
\item  $R^{B1:\text{const}}$ satisfies all five except ECON, 
\item  $R^{B1:\text{NWL}}$ satisfies all five except DCONT,
\item  $R^{B1:\vartriangleright}$ satisfies all five except NT,
\item  $R^{\tilde{B1}}$ satisfies all five except IPP,
\item  $R^\text{const}$ satisfies all five except CU.
\end{itemize}
\end{proposition}

\begin{proof}
The proofs are straightforward, so we will prove the first statement in detail and leave the remaining ones to readers. 

Since the names of the individuals do not matter in the definition of $R^{B1:\text{const}}$, the SRS satisfies NT. Furthermore, regardless of whether $\{x,y\}$ is a buddy pair or not (i.e., regardless of whether (i) or (ii) is applied), the worst equivalence class does not matter in determining $R^{B1:\text{const}}_{\succsim}\mid_{\{x,y\}}$. Hence, IPP also holds. Suppose $\mathord{\succsim}=(\Sigma_{l},\dots,\Sigma_{1})\in\mathcal{L}^{x,y}_{l}$. Since $\lvert\Sigma_{k}\rvert=1$ by assumption, $\{x,y\}$ can appear at most $l$ times among $\Sigma_{l},\dots,\Sigma_{1}$. This means that $\{x,y\}$ is not a buddy pair; so, we have $R^{B1:\text{const}}_{\succsim}\mid_{\{x,y\}}=R^{B1}_{\succsim}\mid_{\{x,y\}}$ by the definition of $R^{B1:\text{const}}$. Hence, CU holds. Next, we will prove that $R^{B1:\text{const}}$ satisfies DCONT. Take any $l\in \mathbb{N}$ with $l\geq 2$, $k_{1}\in[l-1]$, $\mathord{\succsim}=(\Sigma_{l},\dots,\Sigma_{1})\in\mathcal{R}_{l}$, $x,y\in X$, and $\Gamma\in \mathcal{M}(\mathcal{X})$. Suppose that $xP_{\succsim+(\Gamma )_{k_{1}}}y$ and $yP_{\succsim+(\Gamma)_{k_1+1}}x$. By the definition of $R^{B1:\text{const}}$, this can happen only if $\{x,y\}$ is not a buddy pair at either $\mathord{\succsim}+(\Gamma)_{k_{1}}$ or $\mathord{\succsim}+(\Gamma)_{k_{1}+1}$. Hence, we have $R^{B1:\text{const}}_{\succsim+(\Gamma)_{k}}\mid_{\{x,y\}}=R^{B1}_{\succsim+(\Gamma)_{k}}\mid_{\{x,y\}}$ for $k \in \{k_{1},k_{1}+1\}$. Since $R^{B1}$ satisfies DCONT, there exist $\Gamma_{1},\Gamma_{2}\in\mathcal{M}(\mathcal{X})$ such that $\Gamma=\Gamma_{1}+\Gamma_{2}$ and $xI^{B1}_{\succsim+(\Gamma_{1})_{k_{1}}+(\Gamma_{2})_{k_{1}+1}}y$. Notice that shifting this $\Gamma_{2}$ does not make $\{x,y\}$ a buddy pair, and so we can conclude that $xI^{B1:\text{const}}_{\succsim+(\Gamma_{1})_{k_{1}}+(\Gamma_{2})_{k_{1}+1}}y$. This shows that $R^{B1:\text{const}}$ satisfies DCONT. 

Finally, we show that $R^{B1:\text{const}}$ does not satisfy ECON. Let $\mathord{\succsim}=(\mset{\{x,y\}},\mset{\{x,y\}})$ and $\mathord{\succsim}'=(\mset{\{x,y\}},\mset{\{x\}})$. Then, we have $xI^{B1:\text{const}}_{\succsim}y$, $xP^{B1:\text{const}}_{\succsim}y$, and $xI^{B1:\text{const}}_{\succsim+\succsim'}y$. This contradicts ECON.
\end{proof}

\begin{proposition}
The four axioms in Corollary~\ref{corollary:alternative characterization of B1} are logically independent. In fact, 
\begin{itemize}
\item $R^{\tilde{B1}}$ satisfies all four except IPP,
\item $R^{B1,\vartriangleright}$, for any $\vartriangleright$, satisfies all four except NT,
\item $R^\text{const}$ satisfies all four except MON,
\item $R^\text{NWL}$ satisfies all four except CAN.
\end{itemize}
\end{proposition}

\begin{proof}
    Verifying the set of axioms satisfied by each solution is left to the reader; we present counterexamples proving that they do not satisfy the relevant fourth axiom. 

    Let $\mathord{\succsim}_1 = (\mset{\{x,y\}, \{z\}}, \mset{\{x\}, \{y\}})$ and $\mathord{\succsim}_1' = (\mset{\{x,y\}, \{z\}}, \mset{\{x\}, \{y\}, \{x,z\}})$. Since $xI^{\tilde{B_1}}_{\succsim_1}y$, $R^{\tilde{B_1}}$ satisfying IPP would imply $xI^{\tilde{B_1}}_{\succsim_1'}y$, but $xP^{\tilde{B_1}}_{\succsim_1'}y$.

    Let $\mathord{\succsim}_2 = (\mset{\{x\}, \{y\}}, \mset{\{x, y\}})$ and $\vartriangleright = x \vartriangleright y$. We have $xP^{B_1, \vartriangleright}_{\succsim_2}y$, but since $\succsim_2$ is totally symmetric in $x$ and $y$, NT would imply $xI^{B_1, \vartriangleright}_{\succsim_2}y$.

    Let $\mathord{\succsim}_3 = (\mset{\{x,y\}}, \mset{\{x\}, \{y\}})$ and $\mathord{\succsim}_3' = (\mset{\{x\},\{x,y\}}, \mset{\{y\}})$. Since $xI^\text{const}_{\succsim_3}y$, $R^\text{const}$ satisfying MON would imply $xI^\text{const}_{\succsim_3'}y$, but $xI^\text{const}_{\succsim_3'}y$.

    Let $\mathord{\succsim}_4 = (\mset{\{x,y\}}, \mset{\{x\}, \{y\}}, \mset{\{y\}, \{x\}})$ and $\mathord{\succsim}_4' = (\mset{\{x,y\}, \{x\}}, \mset{\{y\}, \{y\}}, \mset{\{x\}})$. Then, $xI^\text{NWL}_{\succsim_4}y$ and $xP^\text{NWL}_{\succsim_4'}y$, which contradicts CAN.
    
\end{proof}

\section{Discussion}
A notable difference between this ``scoring'' solution and the classical Borda welfare function is the sensitivity to score changes. Indeed, in the classical voting framework, \citet{Young1975} shows that any social choice scoring function whose scores are evenly increasing --- heard as: the difference between two consecutive scores is constant and positive --- is the Borda social choice function. It is straightforward to show that this can be extended to social welfare scoring functions: any social welfare scoring function whose scores are evenly increasing is the Borda social welfare scoring function. It is natural to ask if this result still stands in the social ranking framework: is any SRS that gives evenly increasing scores to the equivalence classes of a social ranking the Borda SRS $R^{B_1}$? The answer is no.

Let $l\in\mathbb{N}_{+}$, $\mathord{\succsim}\in\mathcal{R}_l$ and $x\in X$. For $\alpha,\beta\in\mathbb{R}$ with $\beta>0$, we define the score
\[s_{\succsim}^{\alpha, \beta}(x) = \underset{k \in [l]}{\sum} (\alpha + (k-1) \beta) \cdot \underset{S \in m_k : x \in S}{\sum} m_k(S)\]
and the corresponding SRS $R^{\alpha, \beta}_{\succsim}$ as follows, for every $x,y\in X$: 
\[xR^{\alpha, \beta}y \Leftrightarrow  s^{\alpha, \beta}_{\succsim}(x) \geq s^{\alpha, \beta}_{\succsim}(y).\]

Let us first notice that for every strictly positive $\beta\in\mathbb{R}$ and for every $\mathord{\succsim}\in\mathcal{R}_l$, we have $R_{\succsim}^{0,\beta} = R_{\succsim}^{B_1}$, since for every  $x\in X$, $s_{\succsim}^{0,\beta}(x) = \beta \cdot s_{\succsim}^{B_1}(x)$. However, for every $\alpha\in\mathbb{R}\setminus\{0\}$ and every strictly positive $\beta\in\mathbb{R}$, there exists $\mathord{\succsim}\in\mathcal{R}_l$ such that $R_{\succsim}^{\alpha,\beta} \neq R_{\succsim}^{B_1}$. Indeed, by choosing an $\alpha\in\mathbb{R} \setminus \{0\}$, we can construct a solution that does not satisfy IPP. We next provide an example of how, by changing $\alpha$ and $\beta$, we can obtain substantially different rankings of individuals from the same initial social ranking. 
Let us define the following social ranking $\succsim$:
\[\mathord{\succsim} := (\{\{1,3\}, \{2,3\}, \{2\}\}, \{\{1\}, \{2\}, \{2\}, \{3\}\}, \{\{1,3\}, \{1\}, \{1\}, \{1\}, \{3\}\}).\]
Then, we have $2 R_{\succsim}^{B_1} 3 R_{\succsim}^{B_1} 1$, $3 R_{\succsim}^{5,4} 2 R_{\succsim}^{5,4} 1$, and $1 R_{\succsim}^{3,1} 3 R_{\succsim}^{3,1} 2$.

\section{Conclusion}
\label{section:conclusion}

This paper analyzed several definitions of the Borda rule as an SRS and 
indicated the first type,
$R^{B1}$, as the most appropriate, as it satisfies two interesting properties, 
consistency and closeness to unanimity. The idea of $R^{B1}$ is to give $(k-1)$ points to the coalitions in the $k$-th-worst class and rank the individuals by the total points their own coalitions achieve. One of the noteworthy lessons from our results is this scoring assignment. When applied to weak orders allowing indifferences, the Borda rule is usually defined by counting the number of other elements (not the equivalence classes) in voting contexts \citep[see, e.g.,][]{Terzopoulou2021}. Nevertheless, our results show that the usual approach does not guarantee consistency in social ranking contexts (Proposition~\ref{proposition:Borda_rules_and_ECON}). Further investigation of the class of scoring rules could be a future topic of research.

Finally, thanks to Proposition~\ref{proposition:banzhaf}, Theorem~\ref{theorem: characterization of B1} also provides an axiomatic characterization of the ranking $R^{\beta 1}$ according to the Banzhaf value. As a future research direction, it would be interesting to study the axiomatic foundation of the ranking provided by other game-theoretic solutions.

\section*{Acknowledgments}
We thank Michele Aleandri and Meltem \"{O}zt\"{u}rk for helpful discussions. Takahiro Suzuki was supported by JSPS KAKENHI Grant Numbers JP25K01302 and JP26K04808. Stefano Moretti acknowledges support of the ANR project THEMIS (ANR-20-CE23-0018). Rachel Ruellé acknowledges support of the ANR project GATSBII (ANR-24-CE23-6645).

\bibliography{library}

\end{document}